\documentclass{lmcs}
\usepackage{latexsym,xspace,calc,amsthm}
\usepackage{amsmath,amssymb} 

\usepackage{enumitem,mdwlist} 

\makecompactlist{enumeratei*}{enumeratei}

\usepackage{xparse}
\usepackage{tikz}
\usetikzlibrary{calc}
\DeclareDocumentCommand{\hcancel}{mO{0pt}O{1pt}O{0pt}O{-1pt}}{%
    \tikz[baseline=(tocancel.base)]{
        \node[inner sep=0pt,outer sep=0pt] (tocancel) {#1};
        \draw[gray] ($(tocancel.south west)+(#2,#3)$) -- ($(tocancel.north east)+(#4,#5)$);
    }%
}%

\usepackage{microtype}
\usepackage{tgtermes}
\usepackage{fix-cm}

\usepackage{float}
\floatstyle{boxed} \restylefloat{figure}

\usepackage[all]{xy}

\newcommand\card{\myfresh}

\let\myfresh\#
\def\#{\ensuremath{\text{\tt\myfresh}}}

\newcommand\mnat\nat

\newcommand\seqpow\powerset
\newcommand\seqpowi\powerset

\newcommand\dotin{\mathrel{\dot{\in}}}
\newcommand\dotsubseteq{\mathrel{\dot{\subseteq}}}

\newcommand\fv{\f{fv}}

\newcommand\vara{\text{\tt a}}
\newcommand\varb{\text{\tt b}}
\newcommand\varc{\text{\tt c}}
\newcommand\perm{\f{Perm}}

\newcommand\mjg[1]{} 

\newcommand\ON{{\theory{ON}}}
\newcommand\ONM{{\theory{ON}_{\mathfrak M}}}

\newcommand\nat{{\mathbb N}}

\newcommand\plus{{\text{+}}}

\newcommand\finsubseteq{\mathbin{\subseteq_{\text{\it fin}}}}

\newcommand\tbot{{\pmb\bot}}

\newcommand\tin{{\pmb{\in}}}

\renewcommand\land{\wedge}
\renewcommand\lor{\vee}
\newcommand\limp{\Rightarrow}

\usepackage{upgreek}

\usepackage{yfonts}

\usepackage[mathscr]{eucal}

 \renewenvironment{thebibliography}[1]{%
   \begin{odlthebibliography}{#1}%
     \setlength{\parskip}{0ex}%
     \setlength{\itemsep}{3pt}%
     \fontsize{10}{10} 
     \selectfont
}%
 {%
   \end{odlthebibliography}%
 }

\usepackage[colorlinks]{hyperref}
\hypersetup{urlcolor=black,colorlinks}

\usepackage{fancybox}
\newlength{\mylength}
\newenvironment{frameqn}%
{\setlength{\fboxsep}{5pt}
\setlength{\mylength}{\linewidth}%
\addtolength{\mylength}{-2\fboxsep}%
\addtolength{\mylength}{-2\fboxrule}%
\Sbox
\minipage{\mylength}%
\setlength{\abovedisplayskip}{0pt}%
\setlength{\belowdisplayskip}{0pt}%
$$}%
{$$\endminipage\endSbox
{\setlength{\abovedisplayskip}{1pt}%
\setlength{\belowdisplayskip}{0pt}%
\[\fbox{\TheSbox}\]}}

\newenvironment{frametxt}%
{\setlength{\fboxsep}{5pt}
\setlength{\mylength}{\linewidth}%
\addtolength{\mylength}{-2\fboxsep}%
\addtolength{\mylength}{-2\fboxrule}%
\Sbox
\minipage{\mylength}%
\setlength{\abovedisplayskip}{5pt}%
\setlength{\belowdisplayskip}{5pt}%
}%
{\endminipage\endSbox
{\setlength{\abovedisplayskip}{1pt}%
\setlength{\belowdisplayskip}{0pt}%
\[\fbox{\TheSbox}\]}}

\usepackage{graphics}
\usepackage{ifthen}
\usepackage{verbatim}

\newcommand\theory[1]{\ensuremath{\mathsf{#1}}\xspace}

\def\:{{\hspace{-1pt}{:}\hspace{-1.25pt}{:}\hspace{-.5pt}}}
\usepackage{relsize}

\def\id{\mathtxt{id}}

\newcommand\ssm{{{:}\text{=}}}

\newcommand\deffont[1]{{\bf #1}}

\newcommand\mone{{{\text{-}1}}}
\newcommand\liff{\Leftrightarrow}

\newcommand\supp{\f{supp}}

\newcommand\f[1]{\mathit{#1}}

\newcommand\atoms{\ensuremath{\mathbb{A}}\xspace}

\newcommand\dact[1]{}

\newcommand\powerset{\f{pset}}

\newcommand\lmodel{[\hspace{-0.2em}[}
\newcommand\rmodel{]\hspace{-0.2em}]}
\newcommand\model[1]{{\lmodel #1 \rmodel}}

\newcommand\act[0]{{\cdot}}
\newcommand\mact[0]{{{\cdot}{\cdot}}}
\newcommand\is{\text{\ \ \ is\ \ \ }}

\newcommand{\Defiff}
 {\mathrel{\ \ \stackrel{\scriptstyle \mathrm{def}}{\Leftrightarrow}\ \ }}

\newcommand{\defeq}
  {\stackrel{\mathrm{def}}{\,=\,}}

\newcommand\fix{\f{fix}}

\newcommand\ment[0]{\mathrel{\vDash}}

\newcounter{jamieitemcounter}

\newenvironment{thrm}{\begin{thm}}{\end{thm}}
\newenvironment{lemm}{\begin{lem}}{\end{lem}}
\newenvironment{corr}{\begin{cor}}{\end{cor}}
\newenvironment{defn}{\begin{defi}}{\end{defi}}
\newenvironment{nttn}{\begin{nota}}{\end{nota}}
\newenvironment{xmpl}{\begin{exa}}{\end{exa}}
\newenvironment{rmrk}{\begin{rem}}{\end{rem}}
\newtheorem{xrcs}[thm]{Exercise}

\usepackage{stmaryrd}

\newcommand\mathtxt[1]{ \ensuremath{\mathrm{#1}} }

\newcommand\rulefont[1]{\ensuremath{{\mathrm{\bf (#1)}}}}

\newcommand\nontriv{\f{nontriv}}

\newcommand\Forall[1]{\forall #1.}
\newcommand\Exists[1]{\exists #1.}
\newcommand\ExistsUnique[1]{\exists ! #1.}

\def\at{\text{@}}

\usepackage{datetime}

\makeatletter
\@ifpackageloaded{fdsymbol}\@tempswafalse\@tempswatrue
\if@tempswa
  \newcommand{\fdsy@scale}{1.0}
  \newcommand\fdsy@mweight@normal{Book}
  \newcommand\fdsy@mweight@small{Book}
  \newcommand\fdsy@bweight@normal{Medium}
  \newcommand\fdsy@bweight@small{Medium}
  \DeclareFontFamily{U}{FdSymbolA}{}
  \DeclareSymbolFont{fdsymbols}{U}{FdSymbolA}{m}{n}%
  \SetSymbolFont{symbols}{bold}{U}{FdSymbolA}{b}{n}%
  \DeclareFontShape{U}{FdSymbolA}{m}{n}{
      <-7.1> s * [\fdsy@scale] FdSymbolA-\fdsy@mweight@small
      <7.1-> s * [\fdsy@scale] FdSymbolA-\fdsy@mweight@normal
  }{}
  \DeclareFontShape{U}{FdSymbolA}{b}{n}{
      <-7.1> s * [\fdsy@scale] FdSymbolA-\fdsy@bweight@small
      <7.1-> s * [\fdsy@scale] FdSymbolA-\fdsy@bweight@normal
  }{}
  \DeclareMathSymbol{\aleph}{\mathord}{fdsymbols}{"C7}
  \DeclareMathSymbol{\beth}{\mathord}{fdsymbols}{"C8}
  \DeclareMathSymbol{\gimel}{\mathord}{fdsymbols}{"C9}
  \DeclareMathSymbol{\daleth}{\mathord}{fdsymbols}{"CA}
\fi
\makeatother

\begin{document}

\title[Equivariant ZFA]{Equivariant ZFA and the foundations of nominal techniques} 
\thanks{Many thanks for the detailed and constructive comments of two anonymous referees}

\author[M.~Gabbay]{Murdoch Gabbay}

\begin{abstract}
We give an accessible presentation to the foundations of nominal techniques, lying between Zermelo-Fraenkel set theory and Fraenkel-Mostowski set theory, and which has several nice properties including being consistent with the Axiom of Choice.
We give two presentations of equivariance, accompanied by detailed yet user-friendly discussions of its theory and application.
\end{abstract}

\maketitle
\vspace{-1em} 
\tableofcontents

\section{Introduction}

Nominal techniques are based on positing the existence of a set of \emph{atoms} $a,b,c,\ldots\in\atoms$.
These are atomic elements which can be compared for equality but which have few if any other properties. 
The applications of this deceptively simple idea are collectively called \emph{nominal techniques}, and they are surprisingly rich, varied, and numerous.
We list some of them in Subsection~\ref{subsect.app}; just enough to give the reader some flavour of the scope of this field.

Nominal techniques are a success story in the fruitful interaction of logical foundations, mathematics, and computing. 
Yet precisely this interaction means that even experienced readers sometimes struggle to understand what is going on: 
what does it mean when we write `assume a set of atoms', and what does this assumption really buy us? 

We aim to clarify such questions for three types of reader:
\begin{itemize*}
\item
Readers who may have seen nominal techniques in action but have not given much thought to the subtleties involved in making their foundation precise, and who might appreciate an exposition. 
\item
The fellow writer of a paper using nominal techniques, looking for ideas and suggestions on how to set up the foundations.
\item
Experts in mathematical foundations who (perhaps inspired by Section~\ref{subsect.app}) might be interested in equivariance as an interesting new foundational principle for its own sake.
\end{itemize*} 
Looking at the instruction `assume a set of atoms', an expert in foundations may take this to mean that we work in Zermelo-Fraenkel set theory with atoms (\deffont{ZFA}), instead of in Zermelo-Fraenkel set theory (\deffont{ZF}). 
This intuition is correct, albeit not a full picture, so let us start there. 

Now ZFA is unnecessary because atoms could be modelled by $\mathbb N$ in ZF (or by $\powerset(\mathbb N)$ if we want more atoms, and so forth).
In fact ZFA and ZF are equivalent and biinterpretable in the sense that any model of ZFA can be embedded in a model of ZF, and vice-versa, and anything that we express in ZFA can be translated (quite easily) to an assertion about ZF, and vice-versa.
So in terms of expressivity, ZF and ZFA are the same.
And yet:
\begin{itemize*}
\item
if the translation from ZFA to ZF leads to, say, a quadratic increase in proof-size, or 
\item
if ZFA provides an environment in which it is easier to express ourselves, or 
\item
if ZFA is an environment which naturally lets us perceive \emph{native ZFA concepts}\footnote{\dots meaning concepts that are hard to address in full generality in ZF, where we do not give ourselves atoms, but easy to see in ZFA, where we do.  We will see two examples in this paper: equivariance (see Remark~\ref{rmrk.equivar.native}); and freshness and support (see Subsection~\ref{subsect.fresh.well-behaved}); and Subsection~\ref{subsect.app} is, in a sense, a longer list of native ZFA concepts either combined with and extending familiar concepts from ZF, or leading to new constructs and ways of thinking.} 
\end{itemize*}
then the net gain from working in ZFA can be significant \emph{even if} everything could in principle be compiled back down to ZF.

An analogy: Roman numerals I, II, III, IV, 
can express numbers just as arabic numerals 0, 1, 2, 3, 4, but nobody seriously asserts they are functionally equivalent.
There is such a thing as `a good foundation' and `a poor foundation' for a given task. 
Foundations matter.

So in this paper we will explore specifically and in detail what it means when we write `assume a set of atoms'. 
Specifically, I propose --- perhaps a little provocatively --- that it means we are working in \deffont{equivariant Zermelo-Fraenkel set theory} or \emph{EZFA} for short. 

Now a foundation for nominal techniques has been proposed before, in the paper which started the topic \cite{gabbay:newaas-jv}: Fraenkel-Mostowski set theory (\deffont{FM}), which is an elaboration of ZFA with a \emph{finite support} axiom (see Subsection~\ref{subsect.support}).
We will discuss FM in detail, but we can note in this Introduction that FM is inconsistent with the Axiom of Choice, whereas EZFA has the advantage that EZFA plus Choice (EZFAC) is consistent.

\subsection{Motivations for considering the foundations of equivariance}
\label{subsect.app}

This list, which is nonexhaustive and in no particular order, indicates some of the applications of nominal ideas and why working in an environment in which atoms are explicitly available, can be helpful:
\begin{enumerate}
\item
In EZFA, the notion of finiteness generalises naturally to \emph{orbit-finiteness}; the property a set can have of having finitely many orbits under permutations of atoms (see Definition~\ref{defn.permutation.action} and Subsection~\ref{subsect.on.permutations}).

We can then work with orbit-finite computational structures, such as orbit-finite automata.
See \cite{bojanczyk:auttns,klin:locfcs}.
\item
In a similar vein we can consider nominal generalisations of Kleene Algebras~\cite{gabbay:frenrs,kozen:nomkc,kozen:comink}, thus, classes of languages that include orbits under permutations of atoms. 
\item
Homotopy Type Theory~\cite{hottbook} is a dependent type theory whose types include \emph{paths}; simplifying somewhat, two elements are `equal' when there is a path from one to the other.
There are many semantics for these paths; a nominal semantics provides one of the simplest and cleanest~\cite{coquand:cubtt}.
\item
The native universal algebra\footnote{Logic and models of equational theories, i.e. of sets of equalities between terms.} of EZFA interestingly generalises that of ZF.

We can axiomatise theories with binders, including substitution, first-order logic, and the $\lambda$-calculus \cite{gabbay:capasn-jv,gabbay:stodfo,gabbay:repdul}.\footnote{The informed reader will wonder whether these theories look a bit like cylindric algebras~\cite{henkin:cyla}.  The answer is: yes and no, much as the natural numbers look like atoms, but only at first, and the closer we look the less similar they become.} 

These axiomatisations lead to Stone Dualities between: algebras for first-order logic and the $\lambda$-calculus; and topological spaces with points naturally constructed in EZFA~\cite{gabbay:stodfo,gabbay:repdul}.

Furthermore the general theory of nominal algebra models has a rich structure.
This includes familiar constructs, such as a nice generalisation of the HSP theorem \cite{gabbay:nomahs}, as well as constructions for which mathematics based on ZF has no correspondent~\cite{gabbay:metvil,gabbay:finisn}.
\item
The original motivation for nominal techniques was \emph{nominal abstract syntax}, essentially a generalisation of tree-structured data to include name-binding~\cite{gabbay:newaas-jv,gabbay:fountl,pitts:nomsns}.
Extensive implementations exist, most notably perhaps Nominal Isabelle~\cite{urban:nomrti}.
These implementations are designed to allow us to specify and reason about all the applications listed here, and more --- starting of course with the syntax and operational semantics of logic and programming. 
\item 
Nominal rewriting \cite{gabbay:nomr-jv} is a theory of rewriting (directed equality) that lets us reason on theories with binding and has good properties, such as most general unifiers.

The theory of syntax and meta-syntax in EZFA is itself rich; aside from nominal abstract syntax we have nominal terms and their unification \cite{gabbay:nomu-jv}, nominal rewriting as mentioned above, connections to the simply-typed $\lambda$-calculus and higher-order logic \cite{levy:nomufh-jv,gabbay:pnlthf}, and an equation between infinite streams and meta-languages~\cite{gabbay:metvil}.
\item
The equivariance properties described in Figure~\ref{fig.equivar} and Theorem~\ref{thrm.equivar} are generalisations of $\alpha$-equivalence.
Conversely, the variable-renaming that unobtrusively converts $\forall x.(x{=}x)$ into $\forall y.(y{=}y)$ is a special case of the general set-theoretic principles discussed in this paper.

Thus the equivariance principle of this paper is already known to the reader, in the same sense that ring theory is already known to anybody who has added and multiplied numbers. 
\item
More references are e.g. in the bibliography of \cite{pitts:nomsns}.
\end{enumerate}

\section{The language of sets with atoms}

We establish the syntax and semantics of ZFA set theory.
The reader already familiar with this might prefer to skip straight to Definition~\ref{defn.pi}, which is the first item that is not a standard ZFA definition.

\begin{figure}
$$
\begin{array}{l@{\qquad}l@{\hspace{-6em}}l}
\rulefont{AtmEmpty}&
t\in s\limp s\not\in\atoms
\\
\rulefont{EmptySet}&
t\not\in\varnothing
\\
\rulefont{Extensionality}&
s,s'\not\in\atoms\limp (\Forall{\varb}(\varb\in s\liff \varb\in s'))\limp s=s'
\\
\rulefont{Comprehension}&
s\in \{\vara\in t\mid \phi\} \liff (s\in t\land\phi[\vara\ssm s])
\\
\rulefont{Pair}&
t\in\{s,s'\} \liff (t=s\lor t=s')
\\
\rulefont{Union}&
t\in\bigcup s \liff \Exists{\vara}(t\in\vara\land \vara\in s)
\\
\rulefont{Powerset}&
t\in\powerset(s)\liff t\subseteq s 
\\
\rulefont{Induction}&
(\Forall{\vara}(\Forall{\varb{\in}\vara}\phi[\vara\ssm\varb])\limp\phi) \limp \Forall{\vara}\phi
&\qquad\qquad\fv(\phi)=\{\vara\}
\\
\rulefont{Infinity}&
\Exists{\varc}\varnothing\in\varc\land\Forall{\vara}\vara\in\varc\limp\vara{\cup}\{\vara\}\in\varc
\\
\rulefont{AtmInf}&
\neg(\atoms\finsubseteq\atoms)
\\
\rulefont{Replacement}&
\Exists{\varb}\Forall{\vara}\vara\in\varb\liff \Exists{\vara'}\vara'\in u\land \vara=F(\vara')
\\
\rulefont{Choice}&
\varnothing\neq (\powerset^*(s)\to s) &\powerset^*\text{ is nonempty powerset}
\end{array}
$$
\caption{Axioms of ZFA(C)}
\label{fig.zfa}
\end{figure}

\begin{figure}
$$
\begin{array}{l@{\qquad}l@{\qquad}l}
\\
\rulefont{Equivar} & \Forall{\vara{\in}\f{Perm}}(\phi\liff \phi[\vara_1\ssm\vara\act\vara_1,\dots,\vara_n\ssm\vara\act\vara_n]) 
&\fv(\phi)=\{\vara_1,\dots,\vara_n\}
\end{array}
$$
\caption{Equivariance axiom of EZFA}
\label{fig.equivar}
\end{figure}

\subsection{Equivariant ZFA}

\begin{defn}
\label{defn.fol}
Assume a countably infinite collection of \deffont{variable symbols} $\vara,\varb,\varc,\dots$.
Let the \deffont{language of sets with atoms} be defined by:
$$
\begin{array}{r@{\ }l}
s,t::=& \vara \mid \{s,t\} \mid \powerset(s) \mid \bigcup s \mid \{\vara\in t\mid\phi\} \mid \varnothing \mid \atoms
\\
\phi ::=& s=t \mid t\in s \mid \bot \mid \phi\land\phi \mid \forall\vara.\phi 
\end{array} 
$$
\end{defn}

\begin{rmrk}
Definition~\ref{defn.fol} defines a language of first-order logic with:
\begin{itemize*}
\item
equality $=$, 
\item
sets membership $\in$,
\item
terms for pairset $\{s,t\}$, union $\bigcup s$, powerset $\powerset(s)$, bounded comprehension $\{s\in t\mid \phi\}$, and
\item
constant symbols for the empty set $\varnothing$ and the set of atoms $\atoms$.
\end{itemize*}
This is a sufficient foundation for mathematics, and we can define:
\end{rmrk}

\begin{defn}
\label{defn.ZFA}
\begin{enumerate}
\item
Write \deffont{ZFA} for axioms \rulefont{AtmEmpty} to \rulefont{Replacement} from Figure~\ref{fig.zfa}. 
\item
Write \deffont{ZFAC} for ZFA plus \rulefont{Choice}.
\item
Write \deffont{EZFA} for ZFA plus \rulefont{Equivar} from Figure~\ref{fig.equivar}.
\item
Write \deffont{EZFAC} for ZFA plus \rulefont{Choice} plus \rulefont{Equivar}.
\end{enumerate}

In full, EZFAC stands for \emph{Equivariant Zermelo-Fraenkel Set Theory with Atoms and Choice}. 
\end{defn}

\begin{nttn}
In Figure~\ref{fig.zfa} and elsewhere we may use standard syntactic sugar.
For example:
\begin{itemize*}
\item
In \rulefont{Powerset} we write $t\subseteq s$ and this is shorthand for $\Forall{\vara}(\vara\in t\limp\vara\in s)$.
\item
In \rulefont{AtmInf} we write $\finsubseteq$.
For a given $y$, the relation $x\finsubseteq y$ is inductively defined to be least such that $\varnothing\finsubseteq y$ and $x\finsubseteq y\land y'\in y$ implies $x\cup\{y'\}\finsubseteq x$.
\item
In \rulefont{Choice}, $\powerset^*(x)$ which is shorthand for the set of nonempty subsets $\{\vara\in\powerset(x) \mid \vara\neq\varnothing\}$.
\item
We will write $\hookrightarrow$ for injective function-sets, where a function-set is an element representing the graph of a function.
So $x\hookrightarrow y$ is the set of graphs of injective functions from $x$ to $y$.
\end{itemize*}
The interested reader can find detailed explanations of how these standard constructs are built up from first principles in e.g. Section~5 of~\cite{johnstone:notlst}.
We will assume these basic set-theoretic conventions henceforth.
See also the encyclopaedic treatment of set theory in~\cite{jech:sett} and in particular see page~250 where ZFA is discussed. 
\end{nttn}

\begin{rmrk}
\label{rmrk.redundancy}
\begin{enumerate}
\item
There is redundancy in Definition~\ref{defn.fol}: 
\begin{itemize*}
\item
We could restrict terms to just be variable symbols $\vara$ (losing term-formers for pairset, powerset, union, comprehension, emptyset, and the set of atoms). 
\item
Because we have \rulefont{Extensionality}, we could in principle restrict predicates by losing the equality predicate-former (just keeping $\in$, $\bot$, $\land$, and $\forall$).\footnote{But, since first-order logic with equality comes with an equality congruence derivation rule, removing equality from the logic would require us to reintroduce the congruence rule as (long and awkward) explicit axioms; so there is no net gain.
We might also wish to switch to Quine atoms; see Remark~\ref{rmrk.quine}.} 
\end{itemize*}
This restricted language would be just as expressive, but more verbose.
\item
The axioms in Figures~\ref{fig.zfa} and~\ref{fig.equivar} are actually axiom-schemes (for all $s$, $t$, $s'$, $\phi$, and $F$).
We could reduce the number of axiom-schemes, in some cases very easily: for instance by rewriting \rulefont{Emptyset} as $\Forall{\vara}\vara\not\in\varnothing$.
In other cases this is harder, in particular rules \rulefont{Induction} and \rulefont{Replacement} contain parameters $\phi$ and $F$ and since first-order logic lacks second-order quantification by design, these rules are inherently schemas.
This is standard for axiomatisations of set theory.
 
We will not be too concerned about this and we just optimise for what seems to be the most readable form.
\end{enumerate}
\end{rmrk}

\begin{rmrk}
We briefly spell out intuitions for the axioms:
\begin{itemize*}
\item
\rulefont{AtmEmpty}\quad 
Atoms have no elements (so they are extensionally equal to the empty set $\varnothing$).
\item
\rulefont{EmptySet}\quad
$\varnothing$ has no elements. 
\item
\rulefont{Extensionality}\quad 
If $x$ and $x'$ are sets and have the same elements, then $x=x'$.
\item
\rulefont{Comprehension}\quad
If $y$ is an element then $\{x{\in}y\mid \phi(x)\}$ is a set. 
\item
\rulefont{Pair}\quad
If $x$ and $x'$ are elements then so is $\{x,x'\}$.
\item
\rulefont{Union}\quad
If $x$ is an element then so is $\bigcup x=\{x''\mid x''\in x'\in x\}$.
\item
\rulefont{Powerset}\quad
If $x$ is an element then so is $\powerset(x)$ its collection of subsets.
\item
\rulefont{Induction}\quad
$\in$ is well-founded.
\item
\rulefont{Infinity}\quad
A set exists that contains $\nat$, so that (by comprehension) $\nat$ is a set.
\item
\rulefont{AtmInf}\quad
This expresses that $\atoms$ is an infinite set.
\item
\rulefont{Replacement}\quad
If $F$ is a function-class and $z$ is a set then $\{F(x)\mid x\in z\}$ is a set.\footnote{Note that \rulefont{Comprehension} implies that the image of a set $z$ under a function-set $f$ is a set: we just write 
$f`(z) = \{f(x)\mid x\in z\} = \{x'\in z\mid \Exists{x} (x,x')\in f\}$.
So the real power of \rulefont{Replacement} is that it tells us that the image of a set under a function-\emph{class}, is a set.}
\item
\rulefont{Choice}\quad
If $x$ is a set then there exists a \deffont{choice function}-set mapping $\powerset^*(x)$ to $x$.
\item
We discuss \rulefont{Equivar} from Figure~\ref{fig.equivar} from Remark~\ref{rmrk.read.equivar} onwards, after we have constructed some necessary machinery.
\end{itemize*}
\end{rmrk}

\begin{rmrk}[Atoms are infinite]
There are other ways to insist that $\atoms$ is infinite. 
A natural axiom would be $\varnothing\neq \mathbb N\hookrightarrow\atoms$, or in words: there exists a sets injection of natural numbers into $\atoms$. 
However, this would also imply that an infinite subset of $\atoms$ can be well-ordered, which is compatible with \rulefont{Choice} from Figure~\ref{fig.zfa} but incompatible with the FM-style \rulefont{Fresh} rule from Remark~\ref{rmrk.EZFA.AC}.

\rulefont{AtmInf} is implied by $\varnothing\neq\mathbb N\hookrightarrow\atoms$ but is more elementary, in that it retains consistency with both \rulefont{Choice} and \rulefont{Fresh}.
\end{rmrk}

\begin{rmrk}[A design alternative: Quine atoms]
\label{rmrk.quine}
\begin{enumerate}
\item
An alternative to \rulefont{AtmEmpty} is to use an axiom \rulefont{AtmQuine} that $a=\{a\}$ (these are called \emph{Quine atoms}). 
In the implementation of Fraenkel-Mostowski set theory in \cite{gabbay:newaas-jv} it was convenient to use Quine atoms because doing so removed the condition $s,s'\not\in\atoms$ from \rulefont{Extensionality} --- \rulefont{Induction} needs slightly adjustment instead, but this is used far less often \cite{gabbay:thesis}.
\item\label{quite.not.subset}
One might prefer Quine atoms because with Figure~\ref{fig.zfa} as written atoms are empty and thus subsets of every element.
This is arguably undesirable (and was technically inconvenient in the implementation), and Quine atoms avoid it. 
\item
Some people are suspicious of Quine atoms because they make the universe non-wellfounded.
This is unnecessary: the non-wellfoundedness that Quine atoms introduce is so mild as to be negligible, and even were this not the case, fears of non-wellfounded sets may sometimes be overblown.
\end{enumerate}
\end{rmrk}

This concludes the presentation of the syntax and axioms of ZFAC.
We will discuss \rulefont{Equivar} in Remark~\ref{rmrk.read.equivar} and Section~\ref{sect.more.on.equivar}.
First, we need to develop some terminology and a theory of denotation:

\subsection{Atoms, sets, and denotations}

\begin{nttn}
\begin{itemize*}
\item
If $a\in\atoms$ then call $a$ an \deffont{atom}.
\item
If $x\not\in\atoms$ then call $x$ \deffont{a set}.  
\item
We call an $x$ that is either an atom or a set, an \deffont{element}.
\end{itemize*}
\end{nttn}

\begin{rmrk}
So:
\begin{itemize}
\item
\rulefont{AtmEmpty} from Figure~\ref{fig.zfa} states that atoms are empty (cf. Remark~\ref{rmrk.quine}(\ref{quite.not.subset})), and 
\item
\rulefont{EmptySet} and \rulefont{Extensionality} imply that the only empty \emph{set} is $\varnothing$.
\end{itemize}
\end{rmrk}

\begin{defn}
\begin{itemize*}
\item
Define \deffont{free variables} $\fv(s)$ and $\fv(\phi)$ as usual.
For example, $\fv(\Forall{\vara}(\vara = \varb))=\{\varb\}$.
\item
If $\fv(s)=\varnothing$ or $\fv(\phi)=\varnothing$ (so $s$ or $\phi$ have no free variables) then call them \deffont{closed}.
\end{itemize*}
\end{defn}

We introduce a standard notation:
\begin{defn}
\label{defn.enrich}
\begin{enumerate}
\item
Suppose $\mathfrak M$ is a model of ZFA.

Enrich the term and predicate language from Definition~\ref{defn.fol} by admitting elements $x\in \mathfrak M$ as constants ($x$ the constant denotes $x$ in $\mathfrak M$; see Definition~\ref{defn.denotation}(\ref{denotation.x})).
\item
Call terms in this new language \deffont{$\mathfrak M$-terms} and \deffont{$\mathfrak M$-predicates}.
\item\label{enrich.pure}
If $s$ or $\phi$ do \emph{not} mention any of these additional contants from $\mathfrak M$, then call them \deffont{pure}. 
\end{enumerate}
\end{defn}

\begin{defn}
\label{defn.denotation}
Suppose $\mathfrak M$ is a model of ZFA and suppose $s$ is a closed $\mathfrak M$-term and $\phi$ is a closed $\mathfrak M$-predicate.
Then define \deffont{denotations} 
$$
\model{s}_{\mathfrak M}
\quad\text{and}\quad 
\mathfrak M\ment\phi
$$ 
by the usual inductive definition.
We give a relevant selection of cases:
\begin{enumerate*}
\item
$\mathfrak M\ment y\in x$ when $y\in_{\mathfrak M}x$; that is, when $(y,x)$ is in the $\in_\mathfrak M$ relation that we assumed when we wrote `suppose $\mathfrak M$ is a model of ZFA' at the start of this Definition. 
\item
$\mathfrak M\ment x=x'$ when $x=x'$.\footnote{So we interpret equality by literal identity.  This is a standard design choice but not a necessary one.  We could introduce an `equality relation' $=_{\mathfrak M}$ instead, but then we would need axioms saying that $\mathfrak M$ cannot distinguish $\mathfrak M$-equal elements.}
\item
$\mathfrak M\ment\Forall{\vara}\phi$ when $\mathfrak M\ment\phi[\vara\ssm x]$ for every $x\in\mathfrak M$.
\item\label{denotation.x}
$\model{x}_{\mathfrak M}=x$ so that our extra constants denote themselves in $\mathfrak M$. 
\item
$\model{\{\vara\in s\mid \phi\}} = \{x\in\model{s}_{\mathfrak M} \mid \mathfrak M\ment\phi[\vara\ssm x]\}_{\mathfrak M}$; such an element exists in $\mathfrak M$, by \rulefont{Comprehension}.
\item
\dots and so forth.
\end{enumerate*} 
\end{defn}

\subsection{Pairs and permutations}

We recall some standard terminology:
\begin{defn}
\label{defn.ordered.pairs}
\begin{enumerate}
\item
The Kuratowski \deffont{ordered pair} of $s$ and $t$ is defined by
$$
(s,t) = \{\{s,t\},\{s\}\}.
$$
Thus if $x,y\in\mathfrak M$ then this is reflected in the model as
$$
(x,y) = \{\{x,y\},\{x\}\} \in\mathfrak M.
$$
\item
A \deffont{function-set} is a set of ordered pairs that defines the graph of a function. 
\end{enumerate}
\end{defn}

\begin{defn}
\label{defn.pi}
\begin{enumerate}
\item
Let `is a \deffont{permutation}' be the assertion `is a function-set on atoms that is bijective'; or more directly: a permutation is a bijection on atoms.

Thus given a model $\mathfrak M$, an \deffont{$\mathfrak M$-permutation} is an element $\pi\in\mathfrak M$ of which the predicate `this element is a function-set that is a bijection on atoms' is valid in $\mathfrak M$.  
\item\label{pi.perm}
Write $\perm$ for the \deffont{set of permutations}.
In symbols:
$$
\perm = \{\pi \in \atoms\to\atoms\mid \text{$\pi$ is a permutation} \} .
$$
\item
Write $\id$ for the \deffont{identity} permutation such that $\id(a)=a$ for all $a$.
\item
Write $\pi'\circ\pi$ for composition, so that $(\pi'\circ\pi)(a)=\pi'(\pi(a))$.
\item
Write $\pi^\mone$ for the inverse of $\pi$, so that $\pi^\mone\circ\pi=\id=\pi\circ\pi^\mone$.
\item
If $a,b{\in}\atoms$ then write $(a\;b)$ for the \deffont{swapping} (terminology from \cite{gabbay:newaas-jv}) mapping $a$ to $b$,\ $b$ to $a$,\ and all other $c$ to themselves, and take $(a\;a)=\id$.
\end{enumerate}
\end{defn}

\begin{rmrk}
Given a model $\mathfrak M$ we obtain an element $\perm_\mathfrak M$ which represents the bijections of $\atoms_{\mathfrak M}$ that are represented in $\mathfrak M$.
There may be bijections of $\atoms_{\mathfrak M}$ that cannot be represented as elements of $\mathfrak M$; these are simply not function-sets inside the model.
This is a standard distinction.

More detail on design choices for $\perm$ are in Subsection~\ref{subsect.on.permutations}.
\end{rmrk}

\begin{defn}[The pointwise action]
\label{defn.permutation.action}
Given a permutation $\pi\in\f{Perm}$ we define a \deffont{pointwise (atoms-)permutation action} $\pi\act x$ by $\in$-induction as follows:
$$
\begin{array}{r@{\ }l@{\quad}l}
\pi\act a=&\pi(a) & a\in\atoms
\\
\pi\act X=&\{\pi\act x \mid x\in X\} & X\text{ a set}
\end{array}
$$
\end{defn}

\begin{rmrk}
\label{rmrk.read.equivar}
We now have the machinery needed to read axiom \rulefont{Equivar} from Figure~\ref{fig.equivar}: given any predicate in the language from Definition~\ref{defn.fol}, \rulefont{Equivar} asserts that validity is unaffected by uniformly permuting atoms in the free parameters of that predicate.

In other words, \rulefont{Equivar} states that atoms can be permuted provided we do so consistently in all parameters.
For instance if we have proved $\phi(a,b,c)$, then 
\begin{itemize*}
\item
taking $\pi=(a\,c)$ we also know by \rulefont{Equivar} that $\phi(c,b,a)$ and 
\item
taking $\pi=(a\,a')(b\,b')(c\,c')$ we also know by \rulefont{Equivar} that $\phi(a',b',c')$, but 
\item
\rulefont{Equivar} does not tell us that $\phi(a,b,a)$; this may still hold, just not by equivariance because no permutation takes $(a,b,c)$ to $(a,b,a)$.
\end{itemize*}
\end{rmrk}

We now come to a rather subtle observation:
\begin{rmrk}[The dual nature of atoms]
\label{rmrk.dual.nature}
\rulefont{Equivar} expresses that atoms have a dual nature: 
\begin{itemize*}
\item
individually, atoms behave like pointers to themselves --- if we think of $a$ as a pointer just to itself\footnote{If atoms are Quine atoms as per Remark~\ref{rmrk.quine} then this is literally true, in the sense that $a=\{a\}$.} --- but
\item
collectively, atoms have the flavour of variables ranging permutatively over the set of all atoms --- in the sense that, for example, to know that some ZFA assertion holds of a particular $a$, $b$, and $c$ is (provided these are all the atoms mentioned in our assertion) \emph{the same thing} as knowing that it holds of any three distinct atoms.\footnote{This too can be made precise: see Subsection~2.6 and Lemma~4.17 of \cite{gabbay:pnlthf}.} 
\end{itemize*}
We will continue this discussion in Remark~\ref{rmrk.already.e}.
\end{rmrk}

\subsection{Permutation is a group action}

We take a moment to note that the permutation action is a group action on the sets universe, and then we will study equivariance in more detail in Section~\ref{sect.more.on.equivar}.

\begin{lemm}
\label{lemm.perm.group.action}
Suppose $\pi,\pi'\in\mathfrak M$ are permutations and $x\in\mathfrak M$ is any element.
Then we have:
$$
\begin{array}{r@{\ }l}
\id\act x=&x
\\
\pi\act(\pi'\act x)=&(\pi\circ\pi')\act x .
\end{array}
$$
\end{lemm}
\begin{proof}
By a routine $\in$-induction on $\mathfrak M$.
\end{proof}

\begin{corr}
\label{corr.pi.univ.bijection}
Suppose $\pi\in\f{Perm}\in\mathfrak M$ is a permutation.
Then the $\pi$-action on $\mathfrak M$ 
$$
x\in\mathfrak M\longmapsto \pi\act x\in\mathfrak M
$$ 
is a bijection. 
\end{corr}
\begin{proof}
From Lemma~\ref{lemm.perm.group.action} noting that $\pi\act(\pi^\mone\act x)=x$.
\end{proof}

\section{Equivariance}
\label{sect.more.on.equivar}

\subsection{Some helpful notation}

\begin{nttn}
\label{nttn.syntax.action}
Suppose $\mathfrak M$ is a model of ZFA.
Recall from Definition~\ref{defn.enrich} that $\mathfrak M$-terms are syntax enriched with elements of $\mathfrak M$ as constants denoting themselves.

Suppose $\phi$ is a $\mathfrak M$-predicate and $s$ and $\pi$ are $\mathfrak M$-terms --- writing `$s$ and $t$' would be more principled, but the reader will soon see why we write the second term $\pi$.

Then define 
$$
\pi\mact\phi
\quad\text{and}\quad \pi\mact s
$$ 
by:
\begin{enumerate}
\item
$\pi\mact\phi$ is that predicate obtained by replacing every free variable $\vara$ in $\phi$ with $\pi\act\vara$, and every $x\in\mathfrak M$ in $\phi$ with $\pi\act x$.
\item 
$\pi\mact s$ is that term obtained by replacing every free variable $\vara$ in $s$ with $\pi\act\vara$, and every $x\in\mathfrak M$ in $s$ with $\pi\act x$.
\end{enumerate}
An inductive definition would be routine to write out.
\end{nttn}

\begin{xmpl}
For example, if $\mathfrak M$ is a model of ZFA and $\pi\in\f{Perm}\in\mathfrak M$ is a permutation, then 
$$
\pi\mact \Forall{\vara}(x\in\vara\land x\not\in\varb)
\is
\Forall{\vara}(\pi\act x\in\vara\land \pi\act x\not\in\pi\act \varb) .
$$
\end{xmpl}

We link Notation~\ref{nttn.syntax.action} to \rulefont{Equivar} from Figure~\ref{fig.equivar}:
\begin{lemm}
\label{lemm.compact.equivar}
Suppose $\phi$ is a pure predicate (so mentions no elements of $\mathfrak M$) and $s$ is a pure term, with free variables $\vara_1,\dots,\vara_n$.
Then $\pi\mact\phi$ and $\pi\mact s$ can be rewritten more explicitly as follows:
$$
\begin{array}{r@{\ }l}
\pi\mact\phi \is& \phi[\vara_1\ssm\pi\act\vara_1,\dots,\vara_n\ssm\pi\act\vara_n] 
\\
\pi\mact s \is& s[\vara_1\ssm\pi\act\vara_1,\dots,\vara_n\ssm\pi\act\vara_n] 
\end{array}
$$
As a corollary, \rulefont{Equivar} from Figure~\ref{fig.equivar} can be rewritten in the following more compact form:
\begin{frameqn}
\hspace{-15em}\rulefont{Equivar}\hspace{6em}
\Forall{\vara{\in}\f{Perm}}(\phi\liff\vara\mact\phi) .
\end{frameqn}
\end{lemm}
\begin{proof}
A fact of syntax.
\end{proof}

\begin{rmrk}
The \rulefont{Equivar} from Lemma~\ref{lemm.compact.equivar} has two advantages over the \rulefont{Equivar} in Figure~\ref{fig.equivar}: it is more compact, and if we extend our language we can cleanly extend the axiom-scheme by fine-tuning how $\mact$ behaves on the new terms.
See Remarks~\ref{rmrk.ezfa.comment} and~\ref{rmrk.ezfa.choice.terms}.

The only disadvantage to the \rulefont{Equivar} from Lemma~\ref{lemm.compact.equivar} is that we do need to define $\mact$ first. 
\end{rmrk}

\subsection{Equivariance was there all along}

\begin{lemm}
\label{lemm.base.equivar}
Suppose $\mathfrak M$ is a model of ZFA(C).
Then:
\begin{enumerate*}
\item
$\mathfrak M\ment \pi\act y\in\pi\act x$ if and only if $\mathfrak M\ment y\in x$.
\item
$\mathfrak M\ment \pi\act x =\pi\act y$ if and only if $\mathfrak M\ment x=y$.
\item
$\mathfrak M\ment \pi\act x \subseteq\pi\act y$ if and only if $\mathfrak M\ment x\subseteq y$.
\end{enumerate*}
\end{lemm}
\begin{proof}
By a routine induction on $\mathfrak M$ using the fact that:
\begin{itemize*}
\item
By Corollary~\ref{corr.pi.univ.bijection} $\pi$ acts bijectively on atoms, and
\item
by construction in Definition~\ref{defn.permutation.action} $\pi$ acts pointwise on sets.
\qedhere\end{itemize*}
\end{proof}

\begin{rmrk}
Note that the $\in$-induction on $\mathfrak M$ in Lemma~\ref{lemm.base.equivar} is required; we cannot obtain this result just from the fact that $\pi$ is bijective.
To see why, let us (wrongly) define $\pi\act x$ by
$\pi\act x = \pi(x)$ if $x\in\atoms$, and
$\pi\act x =\varnothing$ if $x\not\in\atoms$;
then Lemma~\ref{lemm.base.equivar}(1) fails because if $a\in\atoms$ then $a\in\{a\}$ but $\pi(a)\not\in\pi\act\{a\}=\varnothing$ --- but we need to induct (by just one level) in $\mathfrak M$ to detect this.
\end{rmrk}

We can now observe that every model of ZFA or ZFAC \emph{already} satisfies the equivariance axiom-scheme \rulefont{Equivar} from Figure~\ref{fig.equivar}.
The result goes back to~\cite[Theorem~8.1.10]{gabbay:thesis}:
\begin{frametxt}
\begin{thrm}
\label{thrm.equivar.for.free}
\label{thrm.equivar}
If $\mathfrak M$ is a model of ZFA(C) then $\mathfrak M$ is also a model of EZFA(C).
\end{thrm}
\end{frametxt}
\begin{proof}
Examining \rulefont{Equivar} from Figure~\ref{fig.equivar} we see that it suffices to show that for every closed predicate $\phi$ and closed term $s$ (possibly mentioning elements of $\mathfrak M$): 
$$
\begin{array}{r@{\ }l}
\mathfrak M\ment&\phi\liff\pi\mact\phi
\quad\text{and}
\\
\mathfrak M\ment& \pi\act s = \pi\mact s.
\end{array}
$$
We reason by induction on syntax:
\begin{itemize}
\item
\emph{The cases of $t\in s$ and $s=t$.}\quad
From Lemma~\ref{lemm.base.equivar}.
\item
\emph{The case of $\Forall{\vara}\phi$.}\quad
Suppose $\mathfrak M\ment \Forall{\vara}\phi$. 
By Corollary~\ref{corr.pi.univ.bijection} the action of $\pi$ on $\mathfrak M$ is a bijection, so that $\mathfrak M\ment\phi[\vara\ssm\pi^\mone\act x]$ for every $x\in\mathfrak M$.
By inductive hypothesis it follows that $\mathfrak M\ment\pi\act(\phi[\vara\ssm \pi^\mone\act x])$ and thus $\mathfrak M\ment (\pi\mact\phi)[\vara\ssm x]$ for every $x\in\mathfrak M$.
Thus $\mathfrak M\ment\Forall{\vara}(\pi\mact\phi)$ and so $\mathfrak N\ment\pi\mact\Forall{\vara}\phi$ as required.
\item
\emph{The cases of $\varnothing$ and $\atoms$.}\quad
It is clear from the pointwise action in Definition~\ref{defn.permutation.action} that $\mathfrak M\ment\pi\act\varnothing=\varnothing$ and by assumption $\pi$ is a bijection on $\atoms$ so that again from Definition~\ref{defn.permutation.action} that $\mathfrak M\ment \pi\act\atoms=\atoms$. 
\item
\emph{The cases of $\tbot$ and $\land$.}\quad
Routine.
\item
Other cases are no harder; and we know they must work because everything else is definable in terms of $\in$, $=$, $\forall$, $\land$, and $\bot$.
\qedhere\end{itemize}
\end{proof}

A detailed discussion of the significance of Theorem~\ref{thrm.equivar} is in Subsection~\ref{subsect.introduce.equivariance}.

\subsection{An example}

When we define a function in nominal techniques, we almost always want to know that it commutes with the permutation action.
Broadly speaking we have three options:
\begin{itemize*}
\item
Check this by explicit calculations for every function that we define.
\item
Use Theorem~\ref{thrm.equivar}. 
\item
Use the axiom-scheme \rulefont{Equivar}. 
\end{itemize*}

We illustrate these three methods on a trio of simple but characteristic examples:
\begin{lemm}
\label{lemm.equivar.pairs}
\begin{enumerate}
\item
If $x,y\in\mathfrak M$ then
$$
\pi\act (x,y) = (\pi\act x,\pi\act y) .
$$
\item
If $x\in\mathfrak M$ then 
$$
\pi\act\powerset(x) = \powerset(\pi\act x) .
$$ 
\item
If 
$f\in(X\to Y)\in\mathfrak M$ 
then
$\pi\act f\in(\pi\act X\to \pi\act Y)\in\mathfrak M$.

In words: if $f$ is a function-set in $\mathfrak M$ from $X$ to $Y$ then $\pi\act f$ is a function-set in $\mathfrak M$ from $\pi\act X$ to $\pi\act Y$.
\end{enumerate}
\end{lemm}
\begin{proof}
Immediate from Theorem~\ref{thrm.equivar.for.free} or \rulefont{Equivar} as we prefer, for the pure (Definition~\ref{defn.enrich}(\ref{enrich.pure})) syntax 
\begin{itemize*}
\item
$s=(\vara_1,\vara_2)$, giving us $\pi\act (\vara_1,\vara_2)=(\pi\act\vara_1,\pi\act\vara_2)$ immediately,\footnote{One could object here that ordered pairing was not a primitive term-former in the syntax of the language of sets with atoms in Definition~\ref{defn.fol} (though powersets are), so that the pure term $(\vara_1,\vara_2)$ does not exist. 
This is not a serious problem: we could add it; in an implementation there would be facilities to extend syntax with convenient term-formers; we could add unique choice (how to do this is in Remark~\ref{rmrk.ezfa.choice.terms}); or we can use the predicate form of equivariance instead and expand definitions.  
} 
 and similarly
\item
$s=\powerset(\vara)$ and
\item
$\vara \tin (\varb\to\varc)$.
\end{itemize*}
If the reader prefers a predicate form for the first two cases above, we can take pure predicates $\varb = (\vara_1,\vara_2)$ and $\varb = \powerset(\vara)$.
Then we need a bit of trivial reasoning, illustrated for pairing as follows:
$$
\begin{array}{l@{\quad}l}
\varb = (\vara_1,\vara_2) \liff \pi\act\varb = (\pi\act\vara_1,\pi\act\vara_2)
& \text{Equivariance}
\\
(\vara_1,\vara_2) = (\vara_1,\vara_2) \liff \pi\act(\vara_1,\vara_2) = (\pi\act\vara_1,\pi\act\vara_2)
& \text{Instantiate $\varb$}
\\
\pi\act(\vara_1,\vara_2) = (\pi\act\vara_1,\pi\act\vara_2)
& \text{Equality reflexive}
\end{array}
$$
Notice how well this scales with $s$ and $\phi$; there is essentially no extra work to be done here for larger $s$ and $\phi$, because the reasoning is independent of their logical content and depends only on their syntactic size.

For the reader's convenience we now sketch the explicit calculations that this abstract result corresponds to for the three specific cases of this result.
We reason using the pointwise action from Definition~\ref{defn.permutation.action}, and Lemmas~\ref{lemm.perm.group.action} and~\ref{lemm.base.equivar}, and Definition~\ref{defn.ordered.pairs}, as follows: 
$$
\begin{array}[t]{r@{\ }l}
\pi\act(x,y) =& \pi\act\{\{x,y\},\{x\}\}
\\
=&\{\{\pi\act x,\pi\act y\},\{\pi\act x\}\}
\\
=&(\pi\act x,\pi\act y)
\end{array}
\ 
\begin{array}[t]{r@{\ }l}
\pi\act\powerset(x) =&\pi\act\{x'\mid x'\subseteq x\}
\\
=&\{\pi\act x'\mid x'\subseteq x\}
\\
=&\{x' \mid \pi^\mone\act x'\subseteq x\}
\\
=&\{x' \mid x'\subseteq \pi\act x\}
\\
=&\powerset(\pi\act x)
\end{array}
\ 
\begin{array}[t]{r@{\ }l}
\pi\act f 
=& \{\pi\act (x,y) \mid (x,y)\in f\}
\\
=& \{(\pi\act x,\pi\act y) \mid (x,y)\in f\}
\\
&\text{\emph{\dots more reasoning here.}}
\end{array}
$$ 
The simplest case of pairs is straightforward enough, but even for powersets we see some nontrivial reasoning is required, driven by the semantic content of $\powerset$ (e.g. the introduction of subsets and inverse permutations), and not just by the syntactic form `term-former applied to variable' $\powerset(\vara)$.
As $\phi$ and $s$ get more complicated, this reasoning tends to grow, even though we know it must always succeed; in contrast, the use of an equivariance axiom is always simple and direct. 
\end{proof}

\subsection{Equivariance and Choice}
\label{subsect.equivar.and.choice}

\begin{rmrk}
\label{rmrk.Choice.1}
\rulefont{Choice} is not mentioned in Theorem~\ref{thrm.equivar}.
How can Choice be compatible with equivariance; surely making arbitrary choices is inherently non-equivariant?

Not if the choices are made inside $\mathfrak M$:
from Lemma~\ref{lemm.equivar.pairs} (or direct from Theorem~\ref{thrm.equivar}) 
$$
f\in(\powerset^*(x)\to x)
\quad\text{implies}\quad
\pi\act f\in \powerset^*(\pi\act x)\to \pi\act x.
$$
In words: if $f$ is a choice function for $x$ in $\mathfrak M$ then $\pi\act f$ is a choice function for $\pi\act x$ in  $\mathfrak M$.
We just permute atoms pointwise in the choice functions.
\end{rmrk}

\begin{rmrk}
Continuing Remark~\ref{rmrk.Choice.1}, the following statements are consistent with EZFA (and are derivable in EZFAC):
\begin{enumerate*}
\item
``There exists a total ordering on $\mathbb A$''.
\item
``Every set can be well-ordered (even if the set mentions atoms)''.
\end{enumerate*}
More on this in Remark~\ref{rmrk.EZFA.AC}.
\end{rmrk}

\subsection{What is equivariance}
\label{subsect.introduce.equivariance}

We continue the discussion from Remarks~\ref{rmrk.read.equivar} and~\ref{rmrk.dual.nature}:
\begin{rmrk}
\label{rmrk.already.e}
Equivariance appears in this paper twice: 
\begin{itemize*}
\item
as an axiom-scheme \rulefont{Equivar} in Figure~\ref{fig.equivar} and 
\item
as a ZFA theorem in Theorem~\ref{thrm.equivar}.
\end{itemize*}
Each can be derived from the other.
So which should we take as primitive: the axiom-scheme or the Theorem?

\rulefont{Equivar}-the-axiom-scheme is derivable by Theorem~\ref{thrm.equivar}, just like explicit pairset, powerset, union, comprehension, and emptyset terms in Definition~\ref{defn.fol} are derivable (see Remark~\ref{rmrk.redundancy}).
But being derivable is not the same as being useless.
Consider:
\end{rmrk}

\begin{rmrk}[Equivariance is native to ZFA]
\label{rmrk.equivar.native}
Equivariance is what we might call a \emph{native ZFA concept}: it clearly naturally inhabits a ZFA universe. 

What this means in practice is that some readers find equivariance hard to understand, and pointing to equivariance-the-theorem in Theorem~\ref{thrm.equivar} does not help as much as it should because the difficulty is not the mathematics so much as in the background assumptions which the reader brings to the paper.

It seems unintuitive; surely there must be a mistake somewhere: ``We can't just permute elements.  Suppose atoms are numbers: then are you claiming $1<2$ if and only if $2<1$?''.
This is a silly objection from inside ZFA --- where ``Suppose atoms are numbers'' makes no sense, because they are just not --- but from inside ZF, in which atoms \emph{have} to \emph{be} something else that is not atoms, we can see why such questions are asked.
The reader's experience may differ but for me the background assumption that atoms cannot just \emph{be} atoms is, to this day, an impediment to explaining nominal techniques.\footnote{\label{footnote.cat} Maths may be a pure subject, but a good maths \emph{paper} needs to be correct and also \emph{convincing}.  
It may therefore be helpful to understand and document some practical difficulties encountered while communicating nominal techniques --- usually to people whose viewpoint is very much ZF-shaped.

Case in point: while describing nominal techniques to a category-theorist colleague I mentioned $\atoms$ and he asked `but what \emph{are} atoms?'.  I said atoms were atoms; he refused to accept this, and I was beaten back to setting $\atoms$ to $\mathbb N$ just so the conversation could make progress.  
From this and the rest of the discussion I concluded that my category-theoretic colleague --- \emph{even though he was happy to switch toposes in mid-sentence} --- was actually a ZF fundamentalist.  
If it was not translated all the way down to ZF, then it was incomprehensible.
This paper might not cure him, but it provides a better answer to his question than I was able to provide at the time: we are in ZFA.

Another case in point: in a paper on permissive-nominal techniques, in which we assume a set of atoms $\atoms$ and a partition of $\atoms$ into two infinite halves $\atoms^<$ and $\atoms^>$, a referee refused to allow the paper to be published.  The sticking point was whether the partition was computable; in ZF, atoms have to be something else, so it could be that $\atoms=\mathbb N$ and $\atoms^<$ is `accidentally' an uncomputable set.  
We could have asserted that $\atoms^<$ was computable, but this would validate a ZF-style premise that there must necessarily be internal structure to atoms; that is, that atoms are not \emph{really} atoms.  
The problem was solved by assuming $\atoms^<$ and $\atoms^>$ \emph{first} and then defining $\atoms=\atoms^<\cup\atoms^>$.  This seemed to work: the referee was satisfied and the paper published. 
} 

EZFA can serve as a practical device with which to get talking about nominal material.
When asked ``What are atoms?'' we can simply answer ``elements of $\atoms$ in EZFA''; and when asked ``Why equivariance?'' we can simply answer ``Because we are in EZFA, and it is an axiom.''.
And when asked ``What else?'' we can simply answer ``That's it.''.
\end{rmrk}

\begin{rmrk}[Equivariance is a natural axiom]
\label{rmrk.natural.axiom}
If we are working in a theorem-prover, the equivariance axiom-scheme assumes a particular importance.

The problem is that while in principle every instance of the axiom-scheme \rulefont{Equivar} can be derived, in practice the cost of proving all those instances from first principles \`a la Theorem~\ref{thrm.equivar} scales up with the complexity of $\phi$, and it is not trivial to automate.

My PhD thesis contained an implementation of FM set theory inside Isabelle, and in the end this issue of proving countless instances of Theorem~\ref{thrm.equivar} caused the development to stall.

The irony of an implementation of nominal techniques stalling due to the cost of proving renaming lemmas was not lost on me, and I concluded that it was very important inside a nominal theorem-prover that equivariance have constant cost to the user, regardless of the nature of or size of the predicate $\phi$. 

Assuming an axiom \rulefont{Equivar} is an effective way to reduce the cost of renaming to \emph{constant} effort, namely, the cost of invoking the axiom.
If instead implementation requires explicit proofs of equivariance properties in the style of Theorem~\ref{thrm.equivar}, then our implementation needs to make sure that all instances of this Theorem are proved in a fully automated manner, achieving the same effect \emph{as if} the Theorem were an Axiom.
In this sense, we can say that from the point of view of implementation Equivariance is a natural axiom.
\end{rmrk}

\begin{rmrk}
\label{rmrk.after.phd}
In view of the difficulties discussed in 
Remark~\ref{rmrk.natural.axiom} arising from the practical cost of implementing Theorem~\ref{thrm.equivar} in Isabelle,
after my PhD I initiated a `mark~2' axiomatisation of nominal techniques in which equivariance was an axiom-scheme (technically: an Isabelle Oracle).

This contrasts for example with the first-order axiomatisation spun off by Pitts \cite{pitts:nomlfo,pitts:nomlfo-jv} at the same time, which following~\cite{gabbay:newaas,gabbay:newaas-jv} treated equivariance as a theorem in the style of Theorem~\ref{thrm.equivar}.

So it is an interesting aspect of the mark~2 implementation that equivariance was an axiom-scheme in the style of \rulefont{Equivar}.
Where other work has treated equivariance, it has treated it as a theorem in the style of Theorem~\ref{thrm.equivar} (see for example Theorem~8.1.10 of \cite{gabbay:thesis}, Lemma~4.7 of~\cite{gabbay:newaas-jv}, and Proposition~2 of~\cite{pitts:nomlfo-jv}).
A message of the mark~2 implementation is that in practical engineering terms, it may be useful for equivariance to be an axiom.
\end{rmrk}

\begin{rmrk}[Equivariance the practical time-saver]
\label{rmrk.equivar.time}
We take a step back and consider the practical application of equivariance to writing mathematics papers.
Provided the reader accepts Theorem~\ref{thrm.equivar} --- and alas, as outlined in Remark~\ref{rmrk.equivar.native} not all readers do --- it costs the same to write `from \rulefont{Equivar}' as it does to write `from Theorem~\ref{thrm.equivar}': it is all just `equivariance'.

Consider our example Lemma~\ref{lemm.equivar.pairs}.
This has a one-line proof from equivariance which we give first, and a longer proof by explicit calculations which we also write out (though not in full).

In practice, in nominal techniques if we write a definition we will most probably want to prove it equivariant. 
Without an equivariance principle like \rulefont{Equivar}/Theorem~\ref{thrm.equivar}, the cost of these proofs increases roughly linearly with complexity for each definition, and so can rise roughly quadratically for the paper overall.\footnote{This being a mathematics paper we need make no apology for slipping in an idealised model: a paper containing $n$ definitions where the $i$th definition has complexity $i$ for $1\leq i\leq n$ will require $n$ equivariance proofs by concrete calculations each of length roughly $i$ for $1\leq i\leq n$.  We observe that $\Sigma_{1\leq i\leq n} i$ is order $n^2$.}

Nominal techniques were developed to cut development time of formal proofs in theorem-provers, but they are also effective in rigorous but informal proofs --- i.e. for ordinary mathematics.
Equivariance creates value because it clears tangles of lemmas into a single unifying principle and distils their long proofs by calculation down to crisp one-liners.

This is clear even from simple examples: consider the first line of the proof of Lemma~\ref{lemm.equivar.pairs}, and the proof of Lemma~\ref{lemm.support.properties}.

The first practical application of equivariance in rigorous but informal mathematics was in~\cite{gabbay:frelog}; see Section~7.4.
For more recent examples see uses of Theorem~2.13 of~\cite{gabbay:stodfo} and of Theorem~2.3.1 of~\cite{gabbay:repdul}.
The reader can find practical examples of how equivariance can be usefully applied in normal mathematics, in those papers.
\end{rmrk}

\begin{rmrk}
In conclusion:
\begin{itemize}
\item
When we introduced nominal techniques in \cite{gabbay:newaas-jv}, we could have marketed EZFA as a foundation, and not FM.\footnote{This is not a critisism of FM \emph{per se}.  And since the paper was written in collaboration, this is also not a critisism of my coauthor.}
\item
When we write maths papers using nominal techniques we might do well to be explicit that we are working in ZFA or EZFA.
\item
Equivariance can be usefully applied in the background of normal maths papers, and when so used it converts collections of long renaming lemmas into crisp one-liners (just like it can do in a theorem-prover).
\item
To be most successful, a theorem-prover implementation of names should be based on axioms modelled on EZFA, and not on ZF.

If the implementation is a set theory it should probably \emph{be} EZFA; if it is a (simple) type theory then it will need to assume a primitive type of atoms along with suitable axioms including \rulefont{Equivar}, or an oracle, or at least a universally applicable automated tactic.\footnote{Seasoned with axioms to taste of course, such as Replacement or Choice.  The critical point is that equivariance must have constant cost: if we use Theorem~\ref{thrm.equivar} instead of \rulefont{Equivar} then we \emph{must} include an automated tactic such that use of the Theorem becomes functionally equivalent to invoking an axiom \rulefont{Equivar}.
Cost must \emph{not} scale up with the complexity of the predicate $\phi$.}
\end{itemize} 
\end{rmrk}

\subsection{Five ways to not understand equivariance}

We now try to head off some of the ways in which 
equivariance has been misunderstood:
\begin{rmrk}
\label{rmrk.ezfa.comment}
\begin{enumerate}
\item
$\vara_1,\dots,\vara_n$ from Figure~\ref{fig.equivar} must mention \emph{all} the variables mentioned in the predicate $\phi$.

It is not the case that $x=y$ if and only if $\pi\act x=y$ in general; however, it is the case that $x=y$ if and only if $\pi\act x=\pi\act y$, and similarly for $y\in x$.
\item
If we extend our term language with constants, then we need to extend \rulefont{Equivar} sensibly.
 
For instance, if we introduce a constant $c\in\atoms$ and blindly extend the axiom-scheme \rulefont{Equivar} then we will create a contradiction, because it is not the case that $x=c$ if and only if $\pi(x)=c$.

Instead, we should extend the axiom-scheme \rulefont{Equivar} so that --- using the notation from Notation~\ref{nttn.syntax.action} and~\ref{lemm.compact.equivar} --- $\pi\mact c = \pi\act c$.
This is essentially how Notation~\ref{nttn.syntax.action}(2) works.

In other instances we might consider using a variable instead of a constant, or making our constant be a function over the nonequivariant choice of atoms concerned \dots or perhaps this is all a sign that really $c$ belongs in a different datatype, like $\mathbb N$.
\item
\emph{Equivariant} constants where $\pi\act c=c$ is natural and desired, such as $\atoms$, $\mathbb N$, or $0$, are fine.

Most constants of practical interest are equivariant, the only exception being when for convenience we want to add all elements of a model to our language, as we do in Definition~\ref{defn.enrich}.
\item
Theorem~\ref{thrm.equivar} does not imply that if $\mathfrak M$ is a model of ZFA then every $x\in\mathfrak M$ is equivariant.

It implies that every $x\in\mathfrak M$ that we can reference explicitly using a closed pure term, is equivariant.
That is a very different assertion.
\item
Theorem~\ref{thrm.equivar} does not imply that if $\mathfrak M$ is a model of ZFA then every $x\in\mathfrak M$ has finite support (see \cite{gabbay:newaas-jv} or Subsection~\ref{subsect.support}).

Indeed, $\mathfrak M$ can and in general will contain elements with finite support, infinite support, or no sensible notion of support at all.
For examples see Exercise~\ref{exercise.infinite}.
\end{enumerate}
\end{rmrk}

\begin{rmrk}
\label{rmrk.ezfa.choice.terms}
We can extend our term language with a \emph{unique choice} term-former $\iota\vara.\phi$ with an axiom-scheme
$$
(\ExistsUnique{\vara}\phi) \limp \phi[\vara\ssm \iota\vara.\phi] .
$$
The \rulefont{Equivar} axiom-scheme extends smoothly in this case.

However, we do need to be sensible if we extend our term-language with arbitrary choice $\epsilon\vara.\phi$ (also called \emph{Hilbert's epsilon}) with axiom-scheme
$$
(\Exists{\vara}\phi) \limp \phi[\vara\ssm\epsilon\vara.\phi] 
$$
because we can then write a closed term such as $\epsilon\vara.(\vara\in\atoms)\in\atoms$.

In this case we need to extend \rulefont{Equivar} to the language with $\epsilon$ in such a way that any nonequivariant choices made by $\epsilon$ are respected, for instance by setting $\pi\mact\epsilon\vara.\phi = \pi\act\epsilon\vara.\phi$.\footnote{The version of Isabelle/ZF used during my PhD wielded $\epsilon$ with the joy of a toddler with a biro and a white sofa.  The modern version seems more refined, with uses of $\epsilon$ minimised.}
\end{rmrk}

\section{Relative consistency, and freshness}
\label{sect.consistency}

\subsection{Relative consistency of EZFAC}

We prove that EZFAC is consistent relative to ZF.
This is a known result but aside from a sketch in \cite{gabbay:thesis} this has not been spelled out in the nominal literature: 
\begin{defn}
\label{defn.M.on}
Suppose $\mathfrak M$ is a model of ZFA.
Write $\ONM$ for the class of $\mathfrak M$-ordinals, or (assuming $\mathfrak M$ is fixed) just $\ON$ for short. 
This can be taken to be the least collection in $\mathfrak M$ such that:
\begin{enumerate*}
\item
$\ON$ is \deffont{transitive}, meaning that if $\alpha\in\ON$ and $\alpha'\in\alpha$ then $\alpha'\in\ON$.
\item
If $U$ is a transitive subset of $\ON$, then $U\in\ON$.
\end{enumerate*}
\end{defn}

\begin{thrm}
\label{thrm.ezfa.consistent}
\begin{enumerate}
\item
ZFAC is consistent relative to ZFC.
\item
As a corollary, EZFAC is consistent relative to ZFC.
\end{enumerate}
\end{thrm}
\begin{proof}
The corollary follows from part~1 of this result using Theorem~\ref{thrm.equivar}, so we now prove part~1 of this result.

Assume some model $\mathfrak M$ of ZFC.
We will use this to build a model $\mathfrak N$ of ZFAC.
Note that by Theorem~\ref{thrm.equivar} $\mathfrak N$ will also be a model of EZFAC.

The elements of $\mathfrak N$ will be a proper class in $\mathfrak M$, so the rest of this proof will be conducted inside $\mathfrak M$.
 
Let $x$, $y$, and $z$ range over elements of $\mathfrak M$, and let $i,j\in\{0,1\}\in\mathfrak M$, and recall the construction of ordered pairs $(x,y)\in\mathfrak M$ from Definition~\ref{defn.ordered.pairs}.
Define relations $\dotin$ and $\dotsubseteq$ on $\mathfrak M$ as follows:
\begin{itemize*}
\item
$(x,i)\dotin (y,0)$ is false for $i\in\{0,1\}$. 
\item
$(x,0)\dotin (y,1)$ when $(x,0)\in y$. 
\item
$(x,1)\dotin (y,1)$ when $(x,1)\in y$. 
\item
$(x,i)\dotsubseteq (y,j)$ when for every $(z,k)$ with $z\in\mathfrak M$ and $k\in\{0,1\}$, if $(z,k)\dotin(x,i)$ then $(z,k)\dotin(y,j)$.
\end{itemize*}

Then we define a class $\mathfrak N\subseteq\mathfrak M$ inductively as follows:
$$
\begin{array}{r@{\ }l}
\mathfrak N_0 =& \{(n,0) \mid n\in\mathbb N\}
\\
\mathfrak N_{\alpha} =& \{(x,1) \mid (x,1)\dotsubseteq (\bigcup_{\alpha'<\alpha}\mathfrak N_{\alpha'},1) \} \quad \alpha\in\f{ON}
\\
\mathfrak N =&\bigcup_{\alpha\in\f{ON}}\mathfrak N_{\alpha}
\end{array}
$$
Intuitively every element in $\mathfrak N$ is tagged with $0$ or $1$, where $0$ indicates `$\mathfrak N$ believes that I am an atom' and $1$ indicates `$\mathfrak N$ believes that I am a set'.

$\mathfrak N$ with $\dotin$ extends to a model of the syntax from Definition~\ref{defn.fol} as follows:
\begin{itemize*}
\item
$\mathfrak N\ment x\doteq y$ when $x=y$.
\item
$\mathfrak N\ment \Forall{\vara}\phi$ when $\mathfrak N\ment \phi[\vara\ssm x]$ for every $x\in\mathfrak N$.
\item
$\{x,y\}_{\mathfrak N} = (\{x,y\},1)$.
\item
$\powerset_{\mathfrak N}(y) = (\{x\in\mathfrak N \mid x\dotsubseteq y\},1)$.
\item
$\bigcup_{\mathfrak N} x = (\{x'\in\mathfrak N \mid x'\dotin x\},1)$.
\item
$\{\vara\in y\mid \phi\}_{\mathfrak N} = (\{x\dotin y \mid \mathfrak N\ment\phi[\vara\ssm x]\},1)$.
\item
$\varnothing_{\mathfrak N} = (\varnothing,1)$.
\item
$\atoms_{\mathfrak N} = (\{(n,0)\mid n\in\mathbb N\},1)$.
\end{itemize*}
It is routine to check that this satisfies the axioms of ZFAC from Figure~\ref{fig.zfa}.
In particular, $\mathfrak N$ satisfies \rulefont{Choice} because the ambient universe $\mathfrak N$ does --- we just make choices on tagged sets.
\end{proof}

\subsection{On the group of permutations}
\label{subsect.on.permutations}

In Definition~\ref{defn.pi} and \rulefont{Equivar} we took a \emph{permutation} $\pi\in\perm$ to be any bijection on $\atoms$.
In more advanced usage, choosing what to include in $\perm$ becomes an interesting design decision.
Here is a sample of the options:
\begin{enumerate*}
\item
A permutation $\pi$ is \emph{any} bijection on atoms.

This is the simplest option, and it is what we did in Definition~\ref{defn.pi}.

Note that `any' should be interpreted within the model, so it is perfectly compatible for a permutation to be any bijection on atoms, and also all permutations are finitely supported; this is what happens in Fraenkel-Mostowski set theory (and see next point).
\item
A permutation $\pi$ is a bijection on atoms such that $\nontriv(\pi)=\{a{\in}\atoms\mid\pi(a)\neq a\}$ is finite (\dots or countable, and so forth; choose your favourite cardinality restriction).
We call such a $\pi$ \deffont{finitely supported}. 
\item
Fix some total ordering $\leq$ on $\atoms$.
Then a permutation is a bijection on atoms that \emph{respects} $\leq$, meaning that $a\leq b$ if and only if $\pi(a)\leq\pi(b)$.

This notion is particularly important in a computational context, where we may wish to assume that atoms are orderable.
\item
Fix a partition $\atoms=\atoms^<\cup\atoms^>$ where $\atoms^<$ and $\atoms^>$ are larger than finite (we can fill in our favourite cardinality constraint here; e.g. countable).
Then a permutation is a bijection $\pi$ on atoms such that $\nontriv(\pi)\setminus\atoms^>$ is finite (countable).
This is the basis of \emph{permissive-nominal} techniques as used for instance in~\cite{gabbay:pernl-jv} and is helpful for quantifying over nominal terms style unknowns. 
\item
Fix a $\mathbb Z$-indexed family $(S_k\mid k\in\mathbb Z)$ such that $S_k\subseteq\atoms$, and $S_k\subseteq S_{k\plus 1}$ for every $k\in\mathbb Z$, and (writing $\card$ for cardinality) $\card S_k = \card \atoms\setminus S_k = S_{k\plus j}\setminus S_k$ for every $j,k\in\mathbb N$ such that $j\geq 1$.

Then a permutation is a bijection $\pi$ on atoms such that $\nontriv(\pi)\subseteq S_k$ for some $k\in\mathbb Z$.

This notion of permutation has many useful properties: for instance there exist $\pi$ such that $\card\nontriv(\pi)=\card\atoms$ and for any $\pi,\pi'$ we have that $\card(\fix(\pi)\cap\fix(\pi'))=\card\atoms$.
\end{enumerate*} 
There is no one right answer to the question of `What is a permutation?'.
It depends on what we want to accomplish, and how much structure we want to allow atoms to have.

\subsection{Support and freshness}
\label{subsect.support}

Aside from equivariance, another notable feature of nominal techniques is \emph{support} and \emph{freshness}.
We give a concise account tailored to the `classic' nominal case where all permutations are finite:

\subsubsection{The basic definition}

\begin{nttn}
\label{nttn.fix}
If $A\subseteq\atoms$ write 
$$
\fix(A)=\{\pi\in\perm \mid \Forall{a{\in} A}\pi(a)=a\} .
$$
\end{nttn}

Recall the pointwise atoms-permutation action from Definition~\ref{defn.permutation.action}:
\begin{defn}
\label{defn.small.supported}
\begin{enumerate}
\item\label{small.supported.supports}
Say that $K\subseteq\atoms$ \deffont{supports} an element $x$ when $\Forall{\pi{\in}\fix(K)}\pi\act x=x$.

We may write 
$$
K\$x\quad\text{for}\quad\text{`$K$ supports $x$'.}
$$
\item
If $x$ is $K$-supported for some finite $K\subseteq\atoms$ then call $x$ \deffont{(finitely) supported} and say that $x$ has \deffont{finite support}.
\item
Finally we define $K\#x$ --- in words: $K$ is \deffont{fresh for} $x$ --- when
$$
K\#x\quad\text{for}\quad\Exists{K'{\subseteq}\atoms}(K'\$x\land K\cap K'=\varnothing) .
$$
Taking $K=\{a\}$, we might write $\{a\}\#x$ as $a\#x$.
\item\label{x.equivariant}
If $x$ is supported by $\varnothing$, so that $\pi\act x=x$ for every permutation $\pi$, then call $x$ \deffont{equivariant}.
Otherwise call $x$ \deffont{non-equivariant}.
\end{enumerate}
\end{defn}

\begin{rmrk}
\begin{itemize}
\item
Equivariance does not mean that for every model $\mathfrak M$ of (E)ZFAC and every $x\in\mathfrak M$, $x$ must be equivariant in the sense of Definition~\ref{defn.small.supported}(\ref{x.equivariant}); this is simply a false inference.
\item
Furthermore, just because $X\in\mathfrak M$ is equivariant in the sense of Definition~\ref{defn.small.supported}(\ref{x.equivariant}) does not imply that every $x\in X$ is equivariant.
We have:
\end{itemize}
\end{rmrk}

\begin{xrcs}
\label{exercise.infinite}
\begin{enumerate}
\item
Using Definition~\ref{defn.permutation.action}, prove that $\atoms$ is equivariant (that is, $\pi\act\atoms=\atoms$), but every $x\in\atoms$ is non-equivariant.
\item 
Find an $X$ such that $X$ is equivariant, but every $x\in X$ does not have finite support.\footnote{\emph{Hint:}\ order.}
\end{enumerate}
\end{xrcs}

Lemma~\ref{lemm.equivar.in.char} collects the most important corollaries of Definition~\ref{defn.small.supported}:
\begin{lemm}
\label{lemm.equivar.in.char}
\label{lemm.stuff}
Suppose $X\in\mathfrak M$ is a set (thus $X\not\in\atoms$) and $K\subseteq\atoms$ is a set of atoms.
Then:
\begin{enumerate*}
\item 
$X$ is equivariant if and only if 
$\Forall{x}\Forall{\pi{\in}\f{Perm}}\bigl(x{\in} X \liff \pi\act x\in X\bigr)$.
\item
$K\$X$ if and only if 
$\Forall{x}\Forall{\pi{\in}\fix(K)}\bigl(x{\in} X \liff \pi\act x\in X\bigr)$.
\end{enumerate*}
Suppose $A\subseteq\atoms$ supports some element $x$. 
Then:
\begin{enumerate*}
\setcounter{enumi}{2}
\item\label{stuff.fixsupp.fixelt}
If $\pi\in\fix(A)$ then $\pi\act x=x$. 
\item\label{stuff.fixsupp.identical}
If $\pi(a)=\pi'(a)$ for every $a{\in}A$ then $\pi\act x=\pi'\act x$.
\end{enumerate*}
\end{lemm}
\begin{proof}
By routine calculations from Definition~\ref{defn.small.supported}. 
\end{proof}

\begin{prop}
\label{prop.supp}
Suppose $\mathfrak M$ is a model of ZFA and $x\in\mathfrak M$ has finite support.
\begin{enumerate}
\item
If $C,C'\finsubseteq\atoms$ both support $x$, then $C\cap C'$ supports $x$.
\item
The \deffont{support} of $x$ defined by
$$
\supp(x)=\bigcap\{C\finsubseteq\atoms\mid C\$x\}
$$ 
is the unique least finite set of atoms supporting $x$.
\end{enumerate} 
\end{prop}
\begin{proof}
See e.g. Theorem~2.21 of \cite{gabbay:fountl}.
\end{proof}

We can now give a nice illustration of equivariance: 
\begin{lemm}
\label{lemm.support.properties}
\begin{itemize*}
\item
$\pi\act\supp(x)=\supp(\pi\act x)$.
\item
$K\$x$ if and only if $(\pi\act K)\$(\pi\act x)$.
\item
$a\# x$ if and only if $\pi(a)\#(\pi\act x)$.
\end{itemize*}
\end{lemm}
\begin{proof}
From \rulefont{Equivar}/Theorem~\ref{thrm.equivar}.
\end{proof}

\begin{rmrk}
\label{rmrk.EZFA.AC}
We discussed in Subsection~\ref{subsect.equivar.and.choice} how \rulefont{Choice} from Figure~\ref{fig.zfa} is compatible with \rulefont{Equivar} from Figure~\ref{fig.equivar}.
There is an incompatibility, but it is not with \rulefont{Equivar}.

Consider the nominal \deffont{finite support axiom} that every element has small support, which we can write in a natural notation as follows:
$$
\hspace{-15em}\rulefont{Fresh}\hspace{6em}
\Forall{\vara}\Exists{\varb\finsubseteq\atoms}\varb\$\vara .
$$
The name `Fresh' for this axiom goes back to \cite{gabbay:newaas,gabbay:newaas-jv}; there it was used in the specific case that `small = finite', and written $\Forall{\vara}\Exists{\varb}\varb\#\vara$; proving equivalence to the axiom above is not hard.

If we introduce \rulefont{Fresh} as an axiom to be satisfied by all elements, then this conflicts with \rulefont{Choice}.
For example, a choice function from $\powerset^*(\atoms)$ to $\atoms$ will not have finite support.
We discuss this next:
\end{rmrk}

\subsubsection{\rulefont{Fresh} as a well-behavedness property}
\label{subsect.fresh.well-behaved}

\rulefont{Fresh} is a \emph{native ZFA} property: it is best approached from within a ZFA universe.
\rulefont{Fresh} is typically presented in the literature as an axiom, for example:
\begin{itemize*}
\item
axiom \rulefont{Fresh} just before Definition~4.4 in \cite{gabbay:newaas-jv}, 
\item
the \emph{finite support} axiom in Definition~1(ii) in \cite{pitts:nomlfo-jv},
\item
the axiom \rulefont{Fresh} in Figure~2 of \cite{gabbay:fountl},  and
\item 
Definition~2.2 of \cite{pitts:nomsns}
\end{itemize*}
(stated for finite support in \cite{gabbay:newaas-jv,pitts:nomlfo-jv,pitts:nomsns} and for possibly infinite support in \cite{gabbay:fountl}).

However, it is sometimes preferable to present \rulefont{Fresh} not as an axiom of Fraenkel-Mostowski set theory but instead as a \emph{well-behavedness} or \emph{niceness} property of the larger (E)ZFA universe, alongside other well-behavedness properties such as `being countable', `being computable', `being a closed set in a topological space', and so forth.
Thus, from the point of view of EZFA, being supported is just a property that some elements have and others do not.

There are two reasons for this:
\begin{enumerate}
\item
Even if every element we intend to work with will satisfy \rulefont{Fresh}, presenting \rulefont{Fresh} as a well-behavedness property does no harm and may be beneficial.
To see why, consider that no paper on number theory has ever been rejected because 
\begin{quote}
``number theory assumes that everything has to be a number and therefore countable, which is easily contradicted by an easy G\"odel diagonalisation argument''.
\end{quote}
But this is only because we can rely on readers having a background understanding that ``assume arithmetic'' does not necessarily mean ``\dots and reject infinity''. 
Yet by presenting \rulefont{Fresh} as an axiom we open ourselves to giving a precisely analogous impression, and if by chance we mention a non-supported set --- meaning no more by this than if a number theorist mentions `the set of even numbers' --- then we risk giving the appearance of inconsistency. 

This is not necessary: we can start with EZFA, treat \rulefont{Fresh} as a well-behavedness property, and then everything can proceed rigorously, in a single foundation, and at little to no cost. Had we done this in~\cite{gabbay:newaas-jv} it might have saved some trouble.\footnote{Larry Paulson once told me that in the early days of Isabelle he had trouble getting papers on theorem-provers published because referees said ``Computers are finite, so mechanised mathematics can only ever talk about finite sets and this is obviously too limited to be useful''.  Critisisms of nominal techniques in general and \rulefont{Fresh} in particular have often had a similar nature, but this does not mean we should walk into the trap and keep asking for more.} 
\item
We sometimes require non-supported elements, and furthermore being supported is not a hereditary sets property.

For example: two recent papers \cite{gabbay:stodfo,gabbay:repdul} are concerned with collections of sets which are supported, but whose elements do not have support.
These are `nominal-flavoured' papers --- they concern sets with support properties --- but these sets cannot be expressed in pure FM; their native foundation is EZFA(C). 

For a more concrete example: in \cite[Definition~5.3]{gabbay:newaas-jv} we considered a \emph{concretion} function $x\at a$. 
This can be represented in FM, but for technical reasons to do with preserving finite support, it must be represented as a partial function. 
ZFA is not subject to finite support restrictions, so concretion can be total; it is defined as its FM version where its input has support, and simply takes arbitrary values elsewhere. 

Concerning support not being a hereditary property, note that there are perfectly reasonable sets in the ZFA universe that have finite support even though their elements do not have finite support, and conversely there are sets that do not have finite support even though all their elements do.  
For example:
\begin{itemize*}
\item
`The set of all well-orderings of atoms' is supported by $\varnothing$, since if $O$ is a well-ordering on atoms then so is $\pi\act O$.
However, no well-ordering of $\atoms$ has finite support. 
\item
A well-ordering of $\atoms$, write it $\leq$, encoded as a set of ordered pairs so that ${\leq} = \{(m,n)\in\atoms^2 \mid m\leq n\}$, does not have finite support even though every element of it does ($(m,n)$ is supported by $\{m,n\}$).
\item
\cite[Remark~6.14]{gabbay:newaas-jv} considered splitting $\atoms$ into two or more infinite subtypes of atoms.  
Clearly, these partitions do not themselves have finite support.
\end{itemize*}
\end{enumerate}
Even if the reader's next paper uses an FM sets foundation, which would be perfectly reasonable, then this paper may serve as a reference to which to send the reader to explain why FM sets is a sensible foundation, what larger universe it naturally embeds in, and how using nominal techniques is compatible with the Axiom of Choice.

\subsection{PNL to HOL}

Nominal techniques model names using ZFA-style atoms.
Another model of names is functional arguments in higher-order logic (\deffont{HOL}).
An extended study of the relationship between nominal-style names and HOL-style functional arguments is in \cite{gabbay:pnlthf}.
See specifically Subsection~2.6, Lemma~4.17, and the Conclusion (Section~9) of \cite{gabbay:pnlthf}.
In the Conclusions to that paper we wrote: 
\begin{quote}
\emph{The translation \dots can only raise $n$ variables, in order [so] we lose the equivariance (name symmetry) \dots This matters because in losing symmetry we lose what makes nominal techniques
so distinctive. So although we show how to translate a complete `nominal' proof to a
complete `HOL' proof, we also see how the way in which nominal and HOL proofs
are manipulated and combined, are different.}
\end{quote}
In very brief overview, a conclusion of \cite{gabbay:pnlthf} is that a key difference between modelling names as ZFA atoms, and modelling names as functional arguments, is \emph{precisely} that the nominal representation gives us equivariance.

The mathematics in \cite{gabbay:pnlthf} is extensive and outside the scope of this paper, but the sufficiently dedicated reader might find parts of that paper interesting as further reading.

\section{Conclusions}

We have introduced EZFAC (Definition~\ref{defn.ZFA}), explored its relationship and relative consistency with ZFC sets and FM sets (Section~\ref{sect.consistency}), and given our time and consideration to a key item of the nominal toolbox: 
equivariance. 
Definitions and extended discussion are in the body of the paper.

Equivariance formalises that 
\begin{quote}
\emph{atoms are distinguishable, but interchangable}.
\end{quote}
As noted in Remark~\ref{rmrk.dual.nature} this gives atoms a dual character: individually atoms refer to themselves, but collectively they behave like variables, via the action of permutations.
This is a subtle notion which readers frequently struggle with.

In summary:
\begin{enumerate}
\item
We can present equivariance as a theorem of ZFA or FM (as in Theorem~\ref{thrm.equivar}), or as an axiom-scheme (as in \rulefont{Equivar} in Figure~\ref{fig.equivar}).
This is what the `E' in EZFA represents.

In rigorous but informal mathematics --- that is, in the written arguments which appear in a publication such as this paper --- there is no practical difference between writing ``by equivariance (Theorem~\ref{thrm.equivar})'' or ``by \rulefont{Equivar}'' \emph{except} from the point of view of reaching out to readers whose intuitions are ZF-shaped.

If a reader's starting intuitions are informed only by ZF, then this can make Theorem~\ref{thrm.equivar} difficult to perceive.
There may be benefits in clarity to taking atoms by explicit definition to be elements that populate $\atoms$ and that satisfy \rulefont{Equivar}. 
See Remarks~\ref{rmrk.equivar.native} and~\ref{rmrk.equivar.time}.
\item 
In a theorem-prover, the cost of writing ``by equivariance (Theorem~\ref{thrm.equivar})'' scales at least linearly with the complexity of the predicate, and though this can be ameliorated using automated tactics it is hard to make that cost fully disappear for the user.
Using the axiom-scheme reduces the cost of equivariance to~1.
See Remarks~\ref{rmrk.natural.axiom} and~\ref{rmrk.after.phd}.
\item
Nominal techniques introduce several tools, including nominal atoms-abstraction, support conditions, and equivariance.

At least for my style of nominal techniques, this list is in increasing order of practical importance: equivariance is the most important and often-used, followed by support conditions, followed by atoms-abstraction and related constructs.
\item
Fraenkel-Mostowski set theory (FM) is often proposed as a foundation for nominal techniques (this may be done implicitly, via the Schanuel Topos, the category of Nominal Sets, or similar).
Comparing and contrasting them:
\begin{itemize*}
\item
FM does not have \rulefont{Choice} and assumes \rulefont{Fresh} as an axiom.
\item
EZFAC has \rulefont{Choice} and treats \rulefont{Fresh} as a well-behavedness condition. 
\end{itemize*}
FM is a useful and important environment and, to be clear, this paper should not be read as an argument against it. 
Yet anything that can be done in FM can be done easily in EZFAC, and as a sets foundation for nominal techniques, EZFAC merits serious consideration.
It is compatible with \rulefont{Choice}, and it is relevantly and usefully more general: 
for instance \cite{gabbay:stodfo,gabbay:pnlthf} explicitly require an EZFAC, not an FM, foundation. 
See Subsection~\ref{subsect.fresh.well-behaved}.
And even if we use FM, the fact that it embeds in EZFAC is an important observation.
\end{enumerate}

\renewcommand\href[2]{#2}


\hyphenation{Mathe-ma-ti-sche}
\providecommand{\bysame}{\leavevmode\hbox to3em{\hrulefill}\thinspace}
\providecommand{\MR}{\relax\ifhmode\unskip\space\fi MR }
\providecommand{\MRhref}[2]{%
  \href{http://www.ams.org/mathscinet-getitem?mr=#1}{#2}
}
\providecommand{\href}[2]{#2}

\end{document}